\documentclass[letterpaper]{article} % DO NOT CHANGE THIS
\usepackage{aaai2026}  % DO NOT CHANGE THIS
\usepackage{times}  % DO NOT CHANGE THIS
\usepackage{helvet}  % DO NOT CHANGE THIS
\usepackage{courier}  % DO NOT CHANGE THIS
\usepackage[hyphens]{url}  % DO NOT CHANGE THIS
\usepackage{graphicx} % DO NOT CHANGE THIS
\usepackage{natbib}  % DO NOT CHANGE THIS AND DO NOT ADD ANY OPTIONS TO IT
\usepackage{caption} % DO NOT CHANGE THIS AND DO NOT ADD ANY OPTIONS TO IT
\usepackage{algorithm}
\usepackage{algorithmic}

\usepackage{newfloat}
\usepackage{listings}

\DeclareCaptionStyle{ruled}{labelfont=normalfont,labelsep=colon,strut=off} % DO NOT CHANGE THIS
\floatstyle{ruled}
\newfloat{listing}{tb}{lst}{}
\floatname{listing}{Listing}
\usepackage{graphicx}
\usepackage{amsmath}
\usepackage{multirow}
\usepackage{multicol}
\usepackage{amssymb}
\usepackage{amsthm}
\usepackage{booktabs}
\usepackage{algorithm}
\usepackage{algorithmic}
\usepackage[switch]{lineno}
\usepackage{xcolor}
\newcommand{\mechanism}{\mathcal{M}}
\newcommand{\expectation}[2][\pi]{\mathbb{E}_{#1}[#2]}
\newcommand{\localv}[1][S]{v_{| #1}}
\newcommand{\localpi}[1][S]{\pi_{| #1}}
\newcommand{\share}[4]{\phi_{#1}^{#2}(#3,#4)}

\newcommand{\mcp}{\hat{S}}

\newcommand{\previous}[2]{p^{#1}(#2)}

\DeclareMathOperator{\EW}{EW}
\DeclareMathOperator{\SV}{SV}
\DeclareMathOperator{\MC}{MC}
\DeclareMathOperator{\CR}{CR}
\DeclareMathOperator{\SD}{SD}
\DeclareMathOperator{\WVS}{WVS}
\DeclareMathOperator{\EGM}{EVS}

\newtheorem{definition}{Definition}
\newtheorem{example}{Example}
\newtheorem{theorem}{Theorem}
\newtheorem{proposition}{Proposition}
\usepackage{mathtools} 
\newtheorem{observation}{Observation}

\nocopyright

\title{Fair Incentives for Early Arrival in 0-1 Cooperative Games}
\author {
    Yaoxin Ge\textsuperscript{\rm 1},
    Yao Zhang\textsuperscript{\rm 2},
    Dengji Zhao\textsuperscript{\rm 1}
}
\affiliations {
    \textsuperscript{\rm 1}Key Laboratory of Intelligent Perception and Human-Machine Collaboration, ShanghaiTech University\\
    \textsuperscript{\rm 2}Kyushu University\\
    geyx@shanghaitech.edu.cn, zhang@agent.inf.kyushu-u.ac.jp, zhaodj@shanghaitech.edu.cn
}

\usepackage{bibentry}
\begin{document}

\maketitle

\begin{abstract}
Incentives for early arrival (I4EA) was recently proposed for studying online cooperative games. In an online cooperative game, players arrive in an unknown order, and the value increase after each player arrived should be distributed immediately among all the arrived players. Although there is only one arriving order in the game, we also hope that the value distribution is equal to their Shapley value in expectation. To achieve these goals, the early solutions ignored the fairness in each single arriving order. More specifically, an important player may receive nothing in a game, which seems unfair in reality. To combat this, we propose refined fairness in this paper and design new solutions in 0-1 value games. Specifically, we compute the distance of the distribution in each order to the Shapley value and aim to minimize it. We propose a new mechanism called Egalitarian Value-Sharing (EVS) to do so. We also show that the mechanism can maximize the egalitarian welfare among all the players who made contributions.

\end{abstract}

%     % \link{Code}{https://aaai.org/example/code}
%     % \link{Datasets}{https://aaai.org/example/datasets}

\section{Introduction}
% Compared to statistic models, an online model admitting agents plunging into the market at different time, are more suitable for many real-world scenarios where instant and irrevocable sharing should be decided~\cite{simon1945theory}. This attracts abundant literature that focused on online economic markets, including auctions~\cite{ockenfels2006online,bergemann2010dynamic}, matching~\cite{huang2020fully,huang2024online}, hedonic games~\cite{bullinger2023online, flammini2021online} and other topics~\cite{porter2004mechanism,borodin2005online}. Such a motivation could also apply to cooperative games, whose objective is to fairly allocate valuations to a set of collaborative players. In an online scenario, it is impractical to decide the value distribution after the grand coalition is fully established (e.g., when forming a startup, and we cannot even know at which time all potential players have joined).

Compared to statistical models, online models that allow agents to enter the market at different times are more suitable for many real-world scenarios where instant and irrevocable decisions must be made~\cite{simon1945theory}. This has attracted a substantial body of literature focusing on online economic markets, including auctions~\cite{ockenfels2006online,bergemann2010dynamic}, matching~\cite{huang2020fully,huang2024online}, hedonic games~\cite{bullinger2023online,flammini2021online}, and other related topics~\cite{porter2004mechanism,borodin2005online}. Such motivations can also be extended to cooperative games, whose objective is to fairly allocate value among a set of collaborating players. In online scenarios, it is impractical to determine value distribution only after the grand coalition is fully formed (for example, when forming a startup, it is impossible to know in advance when all potential participants will have joined).

% \citet{ge2024incentives} first formalized the model of online cooperative games, and proposed two main concerns in the model: \textit{incentives for early arrival} and \textit{Shapley-fairness}. The first one ensures that players will join the coalition as soon as possible, which is crucial for applicability to avoid the project being stuck by potential players waiting for each other. The second one is a fairness axiom that require the expected share of each player in all possible arriving order is equal to Shapley value, a classic solution concept satisfying most basic intuitive fairness axioms~\cite{shapley1953value,winter2002shapley}. They achieve both targets by decomposing a game into 0-1 monotone games and applying the method of \textit{Rewarding the First Critical Player} (RFC), which only rewards the first player who is indispensable for the created valuation.

\citet{ge2024incentives} were the first to formalize the model of online cooperative games, identifying two primary concerns within this framework: \textit{incentives for early arrival} and \textit{Shapley-fairness}. The former ensures that players are motivated to join the coalition as early as possible, which is crucial for practical applications to prevent the project from stalling due to potential participants waiting for one another. The latter is a fairness axiom requiring that the expected share of each player, averaged over all possible arrival orders, is equal to their Shapley value—a classic solution concept that satisfies most basic and intuitive fairness axioms~\cite{shapley1953value,winter2002shapley}. They achieve both objectives by decomposing the game into 0-1 monotone games and applying the \textit{Rewarding the First Critical Player} (RFC) method, which allocates rewards exclusively to the first player whose participation is indispensable for generating the coalition’s value.

% Although it achieves that axiom of Shapley-fairness, one can still see some direct limitations on the sharing scheme. The most crucial one is that in a specific order, most players who do make a contribution do not receive anything. This drawback cannot be neglected in some applications since only one order will be actually implemented in real world. To conquer this problem, \citet{aziz2025participation} gave up the notion of Shapley-fairness and examined many other fairness axioms including incentives to stay and incentives to participation to force all contributional players to receive enough share.

% Although the RFC method satisfies the axiom of Shapley-fairness, there remain some notable limitations in its sharing scheme. The most significant issue is that, in any specific arrival order, 
% % most players who contribute to the coalition may not receive any reward. 
%  the value allocation can be far from their Shapley value.
% This drawback is particularly problematic in practical settings, as only one arrival order is realized in reality. 
Although the RFC method satisfies the axiom of Shapley fairness, its sharing scheme still exhibits several notable limitations. The most significant issue is that, for any specific arrival order, the resulting value allocation may deviate substantially from the corresponding Shapley value. This drawback poses practical concerns, since in real-world scenarios only a single arrival order is actually realized. To address this issue, we require a property that emphasizes fairness with respect to a single joining order, and propose a new value allocation mechanism based on this property.

% \citet{aziz2025participation} abandoned the requirement of Shapley-fairness and instead explored alternative fairness axioms, such as incentives to stay and incentives to participate, in order to ensure that all contributing players receive an adequate share.

% In this paper, we notice that although RFC is proved to be complete, it is not the unique scheme to satisfy incentives for early arrival and Shapley-fairness. Hence, it suggests that we do not need to sacrifice Shapley-fairness to make the share in each order looks more balanced. Our work tries to answer the following question: \textit{whether we can find a fairer allocation in each possible arriving order without hurting Shapley-fairness?}

Fortunately, we observe that although the RFC scheme is proven to satisfy both incentives for early arrival and Shapley-fairness, it is not the only such mechanism.  This observation suggests that it may not be necessary to sacrifice current properties in order to achieve a more balanced allocation in each possible arrival order. Our work seeks to address the following question: \textit{Can we design a fairer allocation for each possible arrival order without compromising Shapley-fairness?}

% The key reason of a Shapley-fair scheme does not guarantee a balanced share in each order is that, the value shares in every order can always be far away from Shapley value, but without hurting the expected one. Therefore, to restrict these gaps, we introduce the Euclidean (norm-2) distances between Shapley value and shares in all specific orders, which is referred to as \textit{Shapley distances}. With this definition, we can have many ways to enhance Shapley-fairness in specific orders. Here we focus on minimizing the expected distances.

The main reason a Shapley-fair scheme does not guarantee balanced allocations in every order is that, although the expected share matches the Shapley value, the shares assigned in specific orders can deviate significantly from it. To address this, we introduce the Euclidean ($\ell_2$) distance between the Shapley value and the allocation in each specific order, which we refer to as the Shapley distance. This metric enables us to explore various approaches to improve Shapley-fairness at the level of individual orders. In this work, we focus on minimizing the expected Shapley distance across all possible arrival orders.

% One theoretical advantage of choosing the above target is that, when focusing on decomposed 0-1 monotone games, such a target can be transferred to maximizing egalitarian welfare (i.e., the smallest share) among contributional players with an intuitive constraint: players with symmetric contributions will receive larger share if they come earlier. The egalitarian welfare is natural in 0-1 games, as we have an observation that all contributional players are symmetric when the coalition value increases from 0 to 1.

One theoretical advantage of adopting the above objective is that, in the context of decomposed 0-1 monotone games, it can be reformulated as maximizing egalitarian welfare—that is, increasing the minimum share among all contributing players—subject to an intuitive constraint: among players with symmetric contributions, those who arrive earlier should receive a larger share. The notion of egalitarian welfare is particularly natural in 0-1 games, as all contributing players are symmetric in the moment when the coalition value increases from 0 to 1.

% Finally, our contribution can be summarized as follows.
% \begin{itemize}
%     \item We propose a generalized class of policies for online cooperative games that satisfies both incentives for early arrival and Shapley-fairness, which can give non-zero shares to all contributional players compared to previous work.
%     \item We introduce new notions to further evaluate fairness in implementation of specific orders including Shapley distances and egalitarian welfare. We show that in 0-1 games, we have one single mechanism that can maximize both targets under the constraint of players with symmetric contributions taking more share by coming earlier.
% \end{itemize}

Finally, our contributions can be summarized as follows:
\begin{itemize}
\item We propose a generalized class of allocation policies for online cooperative games that satisfy both incentives for early arrival and Shapley-fairness. Unlike previous approaches, our policies can allocate non-zero shares to all contributing players in each order.
\item We introduce new concepts for evaluating fairness in the implementation of specific arrival orders, including Shapley distances and egalitarian welfare. We further show that, in 0-1 games, there exists a single  mechanism that simultaneously maximizes both objectives, subject to the constraint that players with symmetric contributions who arrive earlier receive larger shares.
\end{itemize}

The remainder of the paper is organized as follows. In \textbf{Preliminaries}, we introduce the setting and summarize the theoretical results from~\citet{ge2024incentives}, including the definition of online cooperative games, the formalization of properties, and the RFC mechanism. In \textbf{New Fairness}, we present our new fairness notions for specific orders and further refine them in the context of 0-1 monotone games. In \textbf{The Mechanism for New Fairness}, we describe our proposed class of mechanisms, analyze their properties, and discuss the minimization of the expected Shapley distance. We offer formal proofs of the theoretical results in \textbf{The Proofs}.

\subsection{Related Work}

\subsubsection{Incentives for Early Arrival}
% Besides the literatures~\cite{aziz2025participation, ge2024incentives} we mentioned, \citet{zhang2025incentives} considered incentives for early arrival in the cost-sharing problem. They follow the requirements proposed in~\citet{ge2024incentives} while the valuation function is replaced by the cost function, and they proposed a mechanism solving all cost sharing games. \citet{zhao2025incentives} provides an in-depth discussion of I4EA and highlights its importance in various domains. However, they did not answer the question of how to establish in-order fairness requirements.

In addition to \citet{ge2024incentives}, \citet{zhang2025incentives} investigated incentives for early arrival in the context of cost-sharing problems. Adopting the previous requirements, but with the valuation replaced by a cost function, they designed a mechanism that solves all cost-sharing games.
\citet{aziz2025participation} abandoned the requirement of Shapley-fairness and instead explored other axioms, such as incentives to stay and incentives to participate which ensure that contributing players receive a positive share.
\citet{zhao2025incentives} provided an in-depth discussion of I4EA and emphasizes its significance across a variety of domains. However, they did not address the question of how to formulate in-order fairness requirements.

\subsubsection{Online Coalition Formation} 

% There are literatures concentrating on the online coalition formation, where players arrive sequentially and the question is how to partition the players to maximize the social welfare~\cite{bullinger2023online,flammini2021online}. The utility of agents come from their preferences on the coalition members. \citet{cohen2025egalitarianism} then proposes the requirement of how to maximize the egalitarian welfare, that is, the minimal utility of the participants. \citet{bullinger2025stability} seeks to achieve stable coalition structures online and treat the most common stability concepts based on deviations by single agents and groups of players. Different from these literatures, all players are gathered in one coalition in our setting and we look to determine a reasonable value share over the players.

There is a body of literature focusing on online coalition formation, where players arrive sequentially and the objective is to partition them in a way that maximizes social welfare~\cite{bullinger2023online, flammini2021online}. In these settings, agents derive utility from their preferences over coalition members. \citet{cohen2025egalitarianism} further introduce the challenge of maximizing egalitarian welfare, i.e., the minimum utility among all participants. \citet{bullinger2025stability} aim to achieve stable coalition structures in online environments, addressing common stability concepts based on deviations by individual agents or groups of players. In contrast to these works, our setting assumes that all players form a single coalition, and our goal is to determine a fair value allocation that can incentivize the players to join as soon as possible.

\subsubsection{Online Mechanism Design}

% We also investigate previous works on mechanism design in an online setting. In \citet{2017Dynamic,pavan2014dynamic}, they look to design auction mechanisms for the scenarios where players with private valuations would arrive and change their valuation over time. Another branch is the online fair division such as \citet{elkind2025temporal,aleksandrov2020online,amanatidis2023fair}. In this field, most setting includes a fixed set of players, and the items to be divided arrive sequentially. Comparing to previous work, our main focus is the incentives for early arrival, which guarantees the participation of players.

We also review previous work on mechanism design in online settings. \citet{2017Dynamic} and \citet{pavan2014dynamic} focus on designing auction mechanisms for environments where players with private valuations arrive over time and may change their valuations. Another line of research concerns online fair division, as studied in~\citet{elkind2025temporal, aleksandrov2020online, amanatidis2023fair}, where typically a fixed set of players participate and the items to be allocated arrive sequentially. In contrast to these studies, our work primarily addresses incentives for early arrival, which are crucial to ensuring player participation.

\section{Preliminaries}

We introduce the formal setting of online cooperative games in this section. The main notations follow the initial work by~\citet{ge2024incentives}, and we also give a brief and necessary introduction of their \emph{Rewarding the First Critical Player (RFC)} mechanism.

\subsection{The Model}
A transferable-utility game is denoted by a pair $(N,v)$, where $N$ is the set of players and $v: 2^N \rightarrow \mathbb{R}_{\geq 0}$ is the \emph{valuation function}. For any coalition $S\subseteq N$, $v(S)$ is interpreted as the value that can be produced by $S$, and we assume $v(\emptyset) = 0$. For two players $i,j\in N$, we say $i,j$ are \emph{symmetric} if $\forall S\subseteq N\setminus\{i,j\}$, $v(S\cup\{i\}) = v(S\cup\{j\})$. In this paper, we focus on \emph{0-1 monotone games}\footnote{In~\citet{ge2024incentives}, the authors have shown that a mechanism on 0-1 monotone games can be generalized through the decomposition of games and linear properties can be maintained.} where $v(S)\in\{0,1\}$ for any $S\subseteq N$, and $v(T)\leq v(S)$ for any $T\subseteq S$.

In an online cooperative game, there is an additional element of some order $\pi \in \Pi(N)$, where $\Pi(N)$ is the set of all permutations of $N$. The players join sequentially following the order $\pi$. If player $i$ joins before player $j$ in order $\pi$, we denote it as $i\prec_\pi j$. Moreover, we denote $\previous{\pi}{i}$ as the players in $\pi$ who arrive before $i$ and herself.

% In our setting, only the local information is accessible through the players that have already arrived.
We say a set of players $S$ is a \emph{prefix} of $\pi$, denoted as $S\sqsubseteq \pi$, if the players in $S$ are the first $|S|$ players in $\pi$. Then, the local game on prefix $S$ consists of a local valuation function $\localv:2^S\rightarrow \mathbb{R}_+$ and a local order $\localpi\in\Pi(S)$, where $\localv(T) = v(T)$ for any $T\subseteq S$ and $i\prec_{\localpi} j \Leftrightarrow i\prec_\pi j$ for any $i,j\in S$. Overall, $(v,\pi)$ describes an online game in which a set of players joins following the order $\pi$ and cooperates following the valuation function $v$, while before all players have joined, $(\localv,\localpi)$ describes the local state that only the players in $S$ have arrived. Moreover, in each local state, only the information about the local game is accessible.
Our goal is to design a value-sharing policy that not only determines a value share for $(v,\pi)$, but also includes value shares for any $(\localv,\localpi)$\footnote{We omit the player set in this tuple because the subscript indicates the set of players and, when there is no subscript, the set of players is referred to by default as $N$.}.

\begin{definition}[Mechanism]
\label{def:mechanism} 
A \textbf{value-sharing policy} $\phi$ maps a game $(v,\pi)$ to an $|N|$-tuple of allocations, so that $\share{i}{}{v}{\pi} \geq 0$ is player~$i$'s share of the value, and $\sum_i \share{i}{}{v}{\pi} = v(N)$.

An \textbf{online value-sharing mechanism} 
 is given by a policy $\phi$, so that after the arrival of each prefix $S \sqsubseteq \pi$, each player $i \in S$ gets a (cumulated) share of $\share{i}{}{\localv}{\localpi}$.
\end{definition}

% In this work, we only consider the \textbf{anonymous} mechanism. That is, if in $(v,\pi)$ we replace a player $i$ with an equivalent player $i'$ to create a game $(v',\pi')$, then $\share{i}{}{v}{\pi} = \share{i'}{}{v'}{\pi'}$. 
To keep the players staying in the cooperation, we require the players' shares to be non-decreasing through other players' arrival. Moreover, to prevent the players from delaying their arrival strategically, we require a player's share to be non-increasing when choosing to delay, assuming the order of others to be fixed.

\begin{definition}[OIR]
\label{def:OIR}
    An online mechanism is \textbf{online individually rational (OIR)}  for value function $v$ if for any arrival order $\pi$ and any $T, S \sqsubseteq \pi$ with $T \subseteq S$, we have $\share{i}{}{\localv[T]}{\localpi[T]}\leq \share{i}{}{\localv}{\localpi}$ for every player $i \in T$.
\end{definition}

 \begin{definition}[I4EA]
     \label{def: I4EA}
      An online mechanism is \textbf{incentivizing for early arrival (I4EA)}  
      if for any player $i$, $\share{i}{}{v}{\pi}\geq \share{i}{}{v}{\pi'}$ for all $\pi$ and $\pi'$ such that $\localpi[N \setminus \{i\}] = \pi'_{|N \setminus \{i\}}$
      and $p^{\pi}(i)\subset p^{\pi'}(i)$.
 \end{definition}

Another desired property is fairness. Here we follow the definition of the \emph{Shapley-fair} proposed in~\citet{ge2024incentives}, and in the next section we give our new fairness properties. Shapley-fairness requires that the expected share of a player equals \emph{Shapley value} when $\pi$ is uniformly selected at random. Intuitively, the Shapley value can be regarded as the expectation of a mechanism which rewards each player with her marginal contribution under uniformly selected $\pi$. We give formal definitions of these concepts as follows.

\begin{definition}
    Given $(v,\pi)$, the \textbf{marginal contribution} of a player $i$ is 
    $\MC_i(v,\pi) := v(\previous{\pi}{i}) - v(\previous{\pi}{i}\setminus \{i\}).$
\end{definition}

\begin{definition}[\citet{shapley1953value}]
    Given $v$, the \textbf{Shapley value} of a player $i$ in $(N,v)$ is 
    $$\SV_i(v):=\frac{1}{|N|!}\sum_{\pi\in \Pi(N)}\MC_i(v,\pi).$$
\end{definition}

\begin{definition}[SF]
    Given $v$, an online mechanism is \textbf{Shapley-Fair (SF)} if for any player $i$, 
    $$\frac{1}{|N|!}\sum_{\pi\in \Pi(N)}\share{i}{}{v}{\pi} = \SV_i(v).$$
\end{definition}

\subsection{The Early Solutions}
In~\citet{ge2024incentives}, the authors demonstrate that there exist certain unsolvable games for which no mechanism can simultaneously satisfy SF, OIR, and I4EA. They further characterize the class of solvable games, as detailed in the following, and propose a complete mechanism called \emph{Rewarding the First Critical Player (RFC)}. 
RFC is formalized in Mechanism~\ref{alg: RFC}
Intuitively, in an online 0-1 monotone game, there is at most one player with a marginal contribution of $1$ (referred to as the \emph{marginal player}), while all other players have marginal contributions of $0$. If the value distribution is based solely on marginal contributions, those players who are essential for value creation—termed \emph{critical players}—may strategically delay their actions to become the marginal player. The RFC mechanism addresses this by awarding the value to the first critical player, thereby ensuring I4EA in most games. However, in unsolvable games, there exists a  ``super-complementary'' player who, despite having an individual value of $0$, can become the unique critical player through delayed participation.

\begin{table}[tb]
\centering
\begin{tabular}{|c|c|c|c|}
\hline
Order & \begin{tabular}[c]{@{}l@{}}Critical\\ Players\end{tabular} & \begin{tabular}[c]{@{}l@{}}Marginal\\ Player\end{tabular} & \begin{tabular}[c]{@{}l@{}}RFC\\ Share\end{tabular} \\ \hline
A-B-C & A,B                                                        & B                                                         & A$\leftarrow$1                                                 \\ \hline
A-C-B & A,C                                                        & C                                                         & A$\leftarrow$1                                                 \\ \hline
B-A-C & B,A                                                        & A                                                         & B$\leftarrow$1                                                 \\ \hline
B-C-A & A                                                          & A                                                         & A$\leftarrow$1                                                 \\ \hline
C-A-B & C,A                                                        & A                                                         & C$\leftarrow$1                                                 \\ \hline
C-B-A & A                                                          & A                                                         & A$\leftarrow$1                                                 \\ \hline
\end{tabular}
\caption{The critical players, marginal player and the share determined by RFC of the game in Example~\ref{eg: RFC example}. Each row corresponds to a joining order listed in the first column.}
\label{tb: RFC eg}
\end{table}

\begin{definition}
   Given a 0-1 monotone $v$ and some order $\pi$, if there exists $S\sqsubseteq\pi$ that $v(S)=1$, then the \textbf{marginal player} is some $i\in S$ satisfying
   $$v(\previous{\pi}{i}) = 1 \text{ and }v(\previous{\pi}{i}\setminus\{i\})=0;$$
   and the \textbf{critical players} are defined as
   $$\CR(v,\pi) = \{j\in\pi \mid v(\previous{\pi}{i}\setminus\{j\}) = 0\}.$$
   Also note that the marginal player is the last critical player.
\end{definition}

% \begin{algorithm}[tb]
%     % \renewcommand{\algorithmicrequire}{\textbf{Input:}}
%     % \renewcommand{\algorithmicensure}{\textbf{Output:}}
%     \floatname{algorithm}{Mechanism}
%      \caption{Rewarding the First Critical Player (RFC)}
%     \label{alg: RFC}
%     \begin{algorithmic}[0]
%          \STATE \textbf{Case 1:} For any $S \sqsubseteq \pi$ with $v(S) = 1$:
%     \[
%     \share{i}{}{\localv}{\localpi} = 
%     \begin{cases} 
%         1 & \text{if } i \in \CR(\localpi, \localv) \text{ and } \\
%         & \forall j \in \CR(\localpi,\localv)\setminus\{i\}, i \preceq j \\
%         0 & \text{otherwise.}
%     \end{cases}
%     \]
    
%     \STATE \textbf{Case 2:} For any $S \sqsubseteq \pi$ with $v(S) = 0$:
%     \[
%     \share{i}{}{\localv}{\localpi} = 0 \quad \text{for every } i \in S.
%     \]
%     \end{algorithmic}
% \end{algorithm}

\begin{algorithm}[tb]
    \renewcommand{\algorithmicrequire}{\textbf{Input:}}
    \renewcommand{\algorithmicensure}{\textbf{Output:}}
    \floatname{algorithm}{Mechanism}
     \caption{Rewarding the First Critical Player (RFC)}
    \label{alg: RFC}
    \begin{algorithmic}[1]
    \REQUIRE $(v,\pi)$
    \ENSURE $\phi\leftarrow\{\phi_j\}_{j\in \pi}$
    \STATE Initialize $\phi_j\leftarrow 0,\forall j\in\pi$.
    \IF{$|\CR(v,\pi)|>0$}
        \STATE $i\leftarrow\text{the first player in }\CR(v,\pi)$.
        \STATE $\phi_i\leftarrow 1$.
    \ENDIF
    
    % 

    %      \STATE \textbf{Case 1:} For any $S \sqsubseteq \pi$ with $v(S) = 1$:
    % \[
    % \share{i}{}{\localv}{\localpi} = 
    % \begin{cases} 
    %     1 & \text{if } i \in \CR(\localpi, \localv) \text{ and } \\
    %     & \forall j \in \CR(\localpi,\localv)\setminus\{i\}, i \preceq j \\
    %     0 & \text{otherwise.}
    % \end{cases}
    % \]
    
    % \STATE \textbf{Case 2:} For any $S \sqsubseteq \pi$ with $v(S) = 0$:
    % \[
    % \share{i}{}{\localv}{\localpi} = 0 \quad \text{for every } i \in S.
    % \]
    \end{algorithmic}
\end{algorithm}

% \begin{definition}[RFC]
%     \textbf{Rewarding the First Critcial Player (RFC)} mechanism is defined by the following value-sharing policy:
%     for any $S\sqsubseteq \pi$ with $v(S) = 1$, and player $i\in S$,
%      $$\share{i}{}{\localv}{\localpi} =\left\{\begin{array}{l}
% 1, \text{ if } % \localv(S) = 1 \text{ and } 
% i \in \CR(\localpi, \localv) \text{ and } 
% \\\quad \forall j\in \CR(\localpi,\localv)\setminus\{i\}, % \setminus\{i\}, 
% i\preceq j, \\
% 0, \text { otherwise. }
% \end{array}\right. $$
% For prefix $S$ with $v(S) = 0$, $\phi_i = 0$ for every $i\in S$.
% \end{definition}

\begin{definition}[Solvability]
\label{def: solvable}
   A 0-1 monotone $v$ is \textbf{unsolvable} if there exists a player $i$ such that $v(\{i\}) = 0$ and $\exists S, \{i\} = \{j\in S\mid v(S) = 1 \text{ and } v(S\setminus \{j\}) = 0\}$. A 0-1 monotone $v$ is solvable in contrast.
\end{definition}

It has been pointed out that only on solvable games, there exists a mechanism that is OIR, SF and I4EA. RFC satisfies these properties on all solvable games. Example~\ref{eg: RFC example} illustrates how RFC operates on an unsolvable game.

% For better interpreting our results, we give the Example~\ref{eg: RFC example} which illustrates the concepts mentioned in this section.

\begin{example}
   \label{eg: RFC example}  
   Consider a 0-1 monotone game with $N=\{A,B,C\}$, and $v(S)=1$ only if $\{A,B\}\subseteq S$ or $\{A,C\}\subseteq S$. The marginal player, critical players, and the share of RFC are shown in Table~\ref{tb: RFC eg}. The Shapley value of $A,B,C$ are $4/6,1/6,1/6$ respectively. Also note that this game is unsolvable. $A$ is the ``super-complementary'' player in $\{A,B,C\}$ who can delay to be the unique critical player, e.g., from order $[B,A,C]$ to $[B,C,A]$.
\end{example}

At the end of this section, we provide some useful observations of the critical players and marginal players as follows. Given $(v,\pi)$, let $i$ be the marginal player, then
\begin{observation}
\label{ob: CR unchanged}
    $\forall j\in\previous{\pi}{i}\setminus{\{i\}}$, if $j$ moves her position in order but still arrives before $i$, then the critical players are unchanged.
\end{observation}
\begin{observation}
\label{ob: CR decrease or increase}
    If $i$ delays, the set of critical players becomes a subset of the original one and if $i$ arrives earlier, the set of critical players becomes a superset of the original one until $i$ is no longer the marginal player.
\end{observation}
\begin{observation}
\label{ob: CR sym}
    The critical players are symmetric to each other in local game $\localv[\previous{\pi}{i}]$.
\end{observation}

\section{New Fairness}
Although the property of SF guarantees fairness in an expectation manner, the exact implementation in a specific order may not feel so fair. To overcome this problem, we propose two new considerations. For symmetric players, we ask the mechanism to determine a weakly higher share for the player arriving earlier.

% Although the SF has guaranteed the fairness from an expectation perspective, the question of how to guarantee the fairness in a specific order is still remaining. In this section, we propose 3 new metrics for one-order fairness. The first metric is the Euclidean distance between the share over every order and the Shapley value.

\begin{definition}[MOS]
    Given $v$, a mechanism $\mechanism$ is \textbf{monotone on symmetric players (MOS)} if for any symmetric players $i,j$ and an order $\pi$ where $i\prec_\pi j$
    $$\share{i}{\mechanism}{v}{\pi}\geq \share{j}{\mechanism}{v}{\pi}.$$
\end{definition}

Another key flaw of SF is that the value shares in each possible order can always be far away from Shapley value, without hurting the expected one. Hence, to evaluate these gaps, we introduce the Euclidean distance between the Shapley value and the share in a specific order.

\begin{definition}[SD]
    Given $(v,\pi)$, the \textbf{Shapley distance (SD)} of a mechanism $\mechanism$ on a player $i$ is defined as 
    $$\SD^\mechanism(v,\pi):=\|\SV_i(v) - \share{i}{\mechanism}{v}{\pi}\|^2.$$
\end{definition}

% With the above definition, there are many ways to enhance the fairness in implementation with additional objectives such as minimizing that maximal Shapley distance or minimizing the expected SD. Notice that such metrics can be applied for any general games. In this paper, we will mainly focus on 0-1 games, and we will show that by proposing a new mechanism with another two nature fairness notions for 0-1 games, we can also minimize the expected Shapley distance.

In this paper, our objective is to design mechanisms for 0-1 monotone games that satisfy SF, OIR, I4EA, and MOS, while minimizing expected SD (later we will also show that minimizing SD on every joining order is impossible). We present equivalent conditions for these requirements in the context of 0-1 monotone games, which serve as the foundation for our proposed mechanism.

% We show firstly, to satisfy SF and OIR, a mechanism should  distribute the value only to the critical players. Then, we show that MOS is not guaranteed by I4EA, and MOS requires the monotonicity on the share of critical players. Finally, by studying a family of mechanisms, we find the expected SD is minimized when the minimal share of critical players is maximized.  

\subsection{Implementation in 0-1 Games}
\begin{table}[tb]
\centering

\begin{tabular}{|c|c|c|c|c|}
\hline
Order                            & \begin{tabular}[c]{@{}c@{}}Share\\ to 1st\end{tabular} & \begin{tabular}[c]{@{}c@{}}Share\\ to 2nd\end{tabular} & \begin{tabular}[c]{@{}c@{}}Share\\ to 3rd\end{tabular} & \begin{tabular}[c]{@{}c@{}}Share\\ to 4th\end{tabular} \\ \hline
$\square\square\square\triangle$ & $1-2\epsilon$                                          & $\epsilon$                                             & $\epsilon$                                             & $0$                                                    \\ \hline
$\square\square\triangle\square$ & $1-\epsilon$                                           & $0$                                                    & $0$                                                    & $\epsilon$                                             \\ \hline
$\square\triangle\square\square$ & $1$                                                    & $0$                                                    & $0$                                                    & $0$                                                    \\ \hline
$\triangle\square\square\square$ & $0$                                                    & $1$                                                    & $0$                                                    & $0$                                                    \\ \hline
\end{tabular}
\caption{A share of the game in Example~\ref{eg: I4EA not MOC}. Each row represents joining orders where $\square$ corresponds to one of $A,B,C$ and $\triangle$ corresponds to the null player $D$.  The share of the $(k-1)$th player in the order is listed in the $k$th column.}
\label{tb: I4EA not MOC}
\end{table}

For 0-1 monotone games, we first demonstrate that any mechanism satisfying OIR and SF must allocate the value exclusively to the critical players, and this allocation should occur precisely when the marginal player joins.
\begin{theorem}
\label{thm: share among CR}
    For any mechanism $\mechanism$ satisfying OIR and SF on all solvable games, there is
    $\sum_{j\in \CR(v,\pi)}\share{j}{\mechanism}{v}{\pi} = 1$
    for any solvable $v$ and arbitrary $\pi$. Moreover, $\share{j}{\mechanism}{v}{\pi} = \share{j}{\mechanism}{\localv[\previous{\pi}{i}]}{\localpi[\previous{\pi}{i}]}$ for any $j$ where $i$ is the marginal player.
\end{theorem}

As noted in Observation~\ref{ob: CR sym}, the critical players are symmetric when the marginal player joins. Intuitively, earlier critical players should receive higher rewards. In fact, this forms an equivalent condition for MOS.

\begin{proposition}
\label{prop: MOS=MOC}
     For an OIR and SF mechanism $\mechanism$, it satisfies MOS for any 0-1 monotone $v$ and any order $\pi$ iff $\share{i}{\mechanism}{v}{\pi}\geq \share{j}{\mechanism}{v}{\pi},$ $\forall i,j\in\CR(v,\pi)$ satisfying $i\prec_\pi j$.    
\end{proposition}

% Example~\ref{eg: I4EA not MOC} shows that MOC is not naturally guaranteed by a mechanism satisfying SF, OIR and I4EA. As for EW, the EW of RFC is $0$ as the first critical player takes all, and a trivial improvement is to share a small proportion $\epsilon$ to all other critical players. This motivates us to design a mechanism that satisfies MOC and maximizes the EW.

% The proof of Theorem~\ref{thm: share among CR} and Proposition~\ref{prop: MOS=MOC} are included in the full version of this paper.
Notice that RFC satisfies MOS as only the share of the first critical player is $1$. However, we justify that MOS is not inherently guaranteed by mechanisms that satisfy SF, OIR, and I4EA by Example~\ref{eg: I4EA not MOC}. In the next Section, we offer a family of mechanisms that are OIR, I4EA, SF and MOS. Moreover, the expected SD over all joining orders is minimized when we maximize the minimal share of the critical players.

\begin{definition}[EW]
    Given $(v,\pi)$, the \textbf{egalitarian welfare of critical players (EW)} of a mechanism $\mechanism$ is defined as
    $\EW^{\mechanism}(v,\pi)=\min_{i\in \CR(v,\pi)}\share{i}{\mechanism}{v}{\pi}$.
\end{definition}

\begin{example}
    \label{eg: I4EA not MOC}
    Consider a game containing $N=\{A,B,C,D\}$ where $v(S) = 1$ only if $\{A,B,C\}\subseteq S$ and the share shown in Table~\ref{tb: I4EA not MOC}. To avoid listing all 24 joining orders, we represent $A,B,C$ with $\square$ and $D$ with $\triangle$. This share is obviously not MOS, but it is SF, I4EA and could be output by an OIR mechanism.
\end{example}

\section{The Mechanism for New Fairness}

% Following the Proposition~\ref{prop: MOS=MOC}, an intuitive idea is to assign a series of monotone non-increasing weights $w(k):\mathbb{N}_+ \rightarrow \mathbb{R}_{\geq 0}$ for the $k$th critical player, and determine a weighted average value share. We claim that this idea fails to satisfy I4EA in all solvable games, as the marginal player may delay her arrival to decrease the number of critical players so that her weight increases and the total weight decreases (an example is provided in Example~\ref{eg: WVS}). This also suggests that equally dividing the marginal contribution to all critical players does not work. 

% To fix the above problem, we propose the idea of determining the \emph{minimal critical prefix} in $(v,\pi)$: for the marginal player $i$, we move $i$ one position forward until the player before $i$ is a critical player, and determine the share of $i$ in that order.

Following the proposed properties, an intuitive idea is to share the value equally among the critical players. We claim that this idea fails to satisfy I4EA in all solvable games, since the marginal player may delay to decrease the number of critical players, as shown in the following example. %(an example is provided in Example~\ref{eg: WVS}).

\begin{example}
\label{eg:ES}
    Consider the game where $N=\{A,B,C,D\}$ and $v(S)=1$ only if $\{A,B,C\}\subseteq S$ or $\{A,B,D\}\subseteq S$. Consider the orders $\pi=[C,A,B,D]$ and $\pi'=[C,A,D,B]$. Notice that $B$ is always the marginal player in both orders. However, the sets of critical players in two orders are different. As listed in Table~\ref{tb: WVS}, $C$ is a critical player in $\pi$ buts not in $\pi'$. Therefore in order $\pi$, $B$ can delay to arrive after $D$ so that the number of critical players decreases and the share of $B$ increases from $1/3$ to $1/2$ if we simply make equal shares (abbreviated as ES in the table) among the critical players.
\end{example}

\begin{table}[tb]
\centering
 \renewcommand{\arraystretch}{1.5}
\begin{tabular}{c|ccccccc|}
\cline{2-8}
\multicolumn{1}{l|}{\multirow{2}{*}{}} & \multicolumn{7}{c|}{Joining Orders}                                                                  \\ \cline{2-8} 
\multicolumn{1}{l|}{}                  & \multicolumn{3}{c|}{C-A-{\color{red}B}-D}         & \multicolumn{2}{c|}{C-A-D-{\color{red}B}}   & \multicolumn{2}{c|}{C-D-A-B} \\ \hline
\multicolumn{1}{|c|}{CR}               & C   & A   & \multicolumn{1}{c|}{{\color{red}B}}   & A   & \multicolumn{1}{c|}{{\color{red}B}}   & A             & B            \\ \cline{2-8} 
\multicolumn{1}{|c|}{RFC}              & 1   & 0   & \multicolumn{1}{c|}{0}   & 1   & \multicolumn{1}{c|}{0}   & 1             & 0            \\ \cline{2-8} 
\multicolumn{1}{|c|}{ES}               & 1/3 & 1/3 & \multicolumn{1}{c|}{\color{red}1/3} & 1/2 & \multicolumn{1}{c|}{{\color{red}1/2}} & 1/2           & 1/2          \\ \hline
\multicolumn{1}{|c|}{MCP}              & \multicolumn{3}{c|}{C-A-B}           & \multicolumn{2}{c|}{C-A-B}     & \multicolumn{2}{c|}{C-D-A-B} \\ \cline{2-8} 
\multicolumn{1}{|c|}{EVS}              & 1/3 & 1/3 & \multicolumn{1}{c|}{\color{blue}1/3} & 2/3 & \multicolumn{1}{c|}{\color{blue}1/3} & 1/2           & 1/2          \\ \hline
\multicolumn{1}{|c|}{SV}               & \multicolumn{7}{c|}{A:5/12; B: 5/12; C:1/12; D: 1/12}                                                \\ \hline
\end{tabular}
\caption{Three joining orders of the game in Example~\ref{eg:ES} and~\ref{eg:WVS}. The critical players (CR), result of equally sharing among critical players (ES), minimal critical prefix (MCP), result of EVS (WVS with $w(k)=1$) and the result of RFC are listed respectively. We only show the share of critical players following their joining order. Moreover, we also present the Shapley value (SV) of the game in the last row. }
\label{tb: WVS}
\end{table}

To fix this problem, we propose the idea of determining the \emph{minimal critical prefix} in $(v,\pi)$: for the marginal player $i$, we move $i$ one position forward until the player before $i$ is a critical player, and determine the share of $i$ in that order. 

\begin{definition}
    Given $(v,\pi)$, the \emph{minimal critical prefix} is
    % $$\hat{S}(v,\pi) :=  \arg\min_{\mathclap{\{S\sqsubseteq\pi \mid v(S\cup\{i\})=1\}}} |S| \cup \{i\}$$
    $$\mcp(v,\pi) :=\left(  \arg\min_{\mathclap{\substack{S\sqsubseteq\pi \\ v(S\cup\{i\})=1}}} |S|\right) \cup \{i\}$$
    where $i$ is the marginal player in $(v,\pi)$.
\end{definition}
Notice that the process of determining the minial critical prefix is not equivalent to simply take $i$ together with the agents (weakly) before the penultimate critical player. As we have shown in Observation~$\ref{ob: CR decrease or increase}$, if $i$ is moved forward, the critical players may increase. For instance, we examine the game and two orders in Example~\ref{eg:ES}. In order $\pi$, $B$ cannot maintain being the marginal player by moving her to earlier positions, so that the minimal critical prefix just includes all players (weakly) before $B$. In contrast, in order $\pi'$, $B$ is still the marginal player if moved one step earlier, which finally induces the same minimal critical prefix as that of $\pi$. Intuitively, we should give $B$ the same share in both orders because of the same minimal critical prefix, or $B$ may delay to increase her value share.

Based on the above idea, we now propose our main result, the \emph{Weighted Value-Sharing Mechanism (WVS)}, as summarized in Mechanism~\ref{alg:WVS}. We show that, taking use of the minimal critical prefix, we can improve not only the idea of equally sharing among critical players, but also any policy determined by a non-increasing weight function $w:\mathbb{N}_+ \rightarrow \mathbb{R}_{\geq 0}$, to satisfy I4EA. When $w$ is a constant function, we call it a \emph{Egalitarian Value-Sharing Mechanism (EVS)}. Moreover, in the next Section, we claim that EVS optimizes EW and expected SD. Intuitively, different weight functions impose different importance of players who arrive earlier: if the weight function decreases sharply, we finally give more reward for those joining earlier. In the extreme case where $w(1)\neq 0$ and $w(k) = 0$ for any $k>1$, the corresponding WVS degenerates to the RFC mechanism.

A WVS is parameterized by a weight function $w$, and the idea of WVS is firstly to determine the share of the marginal players in the minimal critical prefix by weighted average, and then to distribute the rest value to other critical players proportional to their weights. When $v$ is unsolvable, this procedure would produce an inefficient share, where the share of the marginal player is less than $1$ and the shares of others are $0$. To maintain efficiency, we give all the value to the marginal player in such cases.

% At the end of WVS, we check if the total share of players is less than $1$, which may only happen when $|\CR(\localv[\mcp],\localpi[\mcp])|>\CR(v,\pi) = 1$. Therefore, the procedure of a WVS can also check whether the given game is solvable in the process. %$v$ is unsolvable and to realize an efficient value share, we let the marginal player take all.
% We give Example~\ref{eg:WVS} to show the share of a WVS.

\begin{algorithm}[htb]
\floatname{algorithm}{Mechanism}
    \caption{Weighted Value-Sharing Mechanism (WVS)}
    \label{alg:WVS}
    \renewcommand{\algorithmicrequire}{\textbf{Input:}}
    \renewcommand{\algorithmicensure}{\textbf{Output:}}
    
    \begin{algorithmic}[1] %[1] enables line numbers
    \REQUIRE $v,\pi$ and a weight function $w$
    \ENSURE value share $\phi\leftarrow\{\phi_j\}_{j\in \pi}$
    \STATE Initialize $\phi_j\leftarrow 0, \forall j\in \pi$.
    \STATE Let $i^*$ be the marginal player, and let $m$ be the number of critical players in $\pi$.
    \STATE Let $S\leftarrow $ be the minimal critical prefix $\mcp(v,\pi)$, and let $m'$ be the number of critical players in $S$.
    \IF{$m'>m=1$}
        \STATE WARNING: ``$v$ is not solvable''.
        \STATE $\phi_i\leftarrow1$ and RETURN.
    \ENDIF
    \STATE Set $\phi_{i^*} \leftarrow \frac{w(m')}{\sum^{m'}_{k= 1}w(k)}$.
    % \STATE Let $m\leftarrow |\CR(v,\pi)|$ be the number of critical players in $\pi$.
    \STATE For each critical player $j$ in $\CR(v,\pi)\setminus\{i\}$, assume $j$ is the $t$-th arrived player in $\CR(v,\pi)\setminus\{i\}$, set $$\phi_j\leftarrow \frac{w(t)}{\sum^{m-1}_{k=1}w(k)}(1-\phi_{i^*}).$$    
    % \FOR{$t\in\{1,...,m'-1\}$}
    %     \STATE Let $j$ be the $t$th critical player in $\CR(v,\pi)$.
    %     \STATE $\phi_j\leftarrow \frac{w(t)}{\sum_{k\in \{1,...,m'-1\}}w(k)}(1-\phi_{i^*})$.
    % \ENDFOR
    \end{algorithmic}
\end{algorithm}

% \begin{algorithm}[htbp]
%     \caption{Weighted Value-Sharing Mechanism (WVS)}
%     \label{alg:WVS}
%     \begin{algorithmic}[1] %[1] enables line numbers
%     \REQUIRE $(v,\pi)$ and weight function $w(k)$
%     \ENSURE $\phi\leftarrow\{\phi_j\}_{j\in \pi}$
%     \STATE Initialize $\phi_j\leftarrow 0, \forall j\in \pi$.
%     \STATE Find the marginal player $i$ and  the minimal critical prefix $\mcp$.
%     \STATE Let $m\leftarrow |\CR(\localv[\mcp],\localpi[\mcp])|$.
%     \STATE $\phi_i \leftarrow \frac{w(m)}{\sum_{k\in [m]}w(k)}$.
%     \STATE Let $m'\leftarrow |\CR(v,\pi)|-1$.
%     \FOR{$t\in[m']$}
%         \STATE Let $j$ be the $t$th critical player in $\CR(v,\pi)$.
%         \STATE $\phi_j\leftarrow \frac{w(t)}{\sum_{k\in [m']}w(k)}(1-\phi_i)$.
%     \ENDFOR
%     \IF{$\sum_{j\in\pi}\phi_j <1$}
%         \STATE WARNING: ``$v$ is not solvable''.
%         \STATE $\phi_i\leftarrow1$.
%     \ENDIF  
%     \end{algorithmic}
% \end{algorithm}

\begin{example}
\label{eg:WVS}
    Consider the same game as described in Example~\ref{eg:ES}, and here we show the results of WVS with a constant weight function, i.e., the EVS mechanism, which are also listed in Table~\ref{tb: WVS}. As we mentioned before, the minimal critical prefix of $[C,A,D,B]$ is $[C,A,B]$ which is the same as the minimal critical prefix of $[C,A,B,D]$. Therefore, the shares of $B$ in both orders are $1/3$ if applying the EVS. For the remaining critical players, they will equally share the rest value since they have equal weights. Now, we can see that the mechanism avoid the strategic delay of player $B$. Moreover, in each order, compared the share of that in RFC, where player $B$ gets nothing, the value distribution is fairer since it is closer to Shapley value.
\end{example}

\subsection{Properties}

We now present the theoretical properties of WVS. In Theorem~\ref{thm: WVS complete}, we show that WVS satisfies SF, OIR, I4EA and MOS if the weight function is weakly monotone decreasing.

\begin{theorem}
\label{thm: WVS complete}
    For any solvable $v$, WVS defined by $w$ is SF, OIR, I4EA and MOS if $w$ is weakly monotone decreasing.
\end{theorem}

Also note that not all mechanisms satisfying these properties can be written as a WVS defined by some specific $w(k)$, because there might exist mechanisms that output different value share for orders with same numbers of critical players. Such mechanisms require a weight system defined on different joining orders, which raises the complexity of mechanism definition and computation to exponential level.  

Now we move to the question that how to optimize the EW and SD. First, we naturally constrain the space of all possible mechanisms into \emph{anonymous} ones. That is, if we replace a player $i$ with an equivalent player $i'$ in $(v,\pi)$ to create a game $(v',\pi')$, then $\share{i}{}{v}{\pi} = \share{i'}{}{v'}{\pi'}$. 
Surprisingly, we find that the EVS, i.e, a WVS with the constant weight function optimizes EW and the expected SD. %We call the WVS defined by a constant weighted function the \emph{Egalitarian Value-Sharing Mechanism (EVS)}.

\begin{definition}
    The Egalitarian Value-Sharing Mechanism is an instance of the WVS mechanism, where the weight system is given by $w(k)=c$ for some constant $c>0$.
\end{definition}

\begin{theorem}
\label{thm: EGM max EW}
    Given any solvable $v$ and joining order $\pi$, $$\EW^{\EGM}(v,\pi)\geq \EW^{\mechanism}(v,\pi),$$ where $\mechanism$ is an arbitrary mechanism satisfying SF, OIR, I4EA and MOS.
\end{theorem}

\begin{theorem}
\label{thm: EGM min d to SV}
     Given any solvable $v$, there is
     $$\expectation[\pi\sim\mathcal{U}(\Pi(N))]{\SD^{\EGM}(v,\pi)}\leq \expectation[\pi\sim\mathcal{U}(\Pi(N))]{\SD^\mechanism(v,\pi)}$$
      where $\mechanism$ is an arbitrary anonymous mechanism satisfying SF, OIR, I4EA and MOS.
\end{theorem}

We further argue that no mechanism can optimize SD for every single $(v,\pi)$. Although critical players may initially be symmetric when the marginal player joins, this symmetry can be disrupted as additional players arrive. For example, in Example~\ref{eg:WVS}, players $A$, $B$, and $C$ are symmetric in the local order $[A,B,C]$, but this symmetry no longer holds in the order $[A,B,C,D]$. Due to such asymmetry, a mechanism that optimizes SD would ideally allocate shares to critical players proportionally to their Shapley values. However, the Shapley value itself may change as new players join, while any value allocation becomes irrevocable once it is made.
\section{The Proofs}
\label{sec: proofs}

Now we offer formal proofs of the theorems and propositions in this paper. For simplicity, we denote $\pi=[...,i,...,j,...]$ and $\pi'=[...,j,...,i,...]$ as a pair of joining orders where only the positions of $i$ and $j$ are exchanged and the order of others is fixed. When we refer to $i,j$ as adjacent players, we use $\pi=[...i,j...]$ (no ellipsis between $i,j$).

\subsection{Sharing among Critical Players}
% We give the proofs of Theorem~\ref{thm: share among CR} and Proposition~\ref{prop: MOS=MOC} as follows, which serves as the foundation for our mechanism.

\begin{proof}[Proof of Theorem~\ref{thm: share among CR}]

For a SF mechanism, the value should be shared to the players entirely at any time. On 0-1 monotone games, the value becomes $1$ when the marginal player joins and keeps unchanged. Therefore, the share should be determined at that time and due to OIR, it should not be changed. As for sharing only on critical players, we prove by induction: given solvable $v$, the equality holds if for any $i\in N$ and any $\localpi[N\setminus\{i\}]\in \Pi(N\setminus\{i\})$, only players in $\CR(\localv[N\setminus\{i\}],\localpi[N\setminus\{i\}])$ receive the value share in the local game $(\localv[N\setminus\{i\}],\localpi[N\setminus\{i\}])$. For simplicity, let $S_i = N\setminus\{i\}$. 

For $\pi \in \Pi(N)$, let $i$ be the last player in $\pi$. We discuss by categories:
\begin{itemize}
    \item If $i$ is not the marginal player, then $\phi_j(v,\pi) = \phi_j(\localv[S_i],\localpi[S_i])$ following OIR, and only critical players have a positive share as $\CR(v,\pi) = \CR(\localv[S_i],\localpi[S_i])$. 
    \item If $i$ is the marginal player, then $i$ is also a critical player whenever $i$ is not the last to join. Moreover, following SF, the value share should satisfy
    \begin{align*}
        &\sum_{\pi\in \Pi(N)}\share{i}{}{v}{\pi} - \sum_{j\in S_i}\sum_{\localpi[S_j]\in\Pi[S_j]}\share{i}{}{\localv[S_j]}{\localpi[S_j]} \\
        =&\sum_{\pi\in \Pi(N)}\MC_i(v,\pi) - \sum_{j\in S_i}\sum_{\localpi[S_j]\in\Pi[S_j]}\MC_i(\localv[S_j],\localpi[S_j]) \\
        =&\sum_{\localpi[S_i]\in \Pi(S_i)}\MC_i(v,\pi).
    \end{align*}
    i.e., over all joining orders, the total value share increase of $i$  equals her total marginal contribution. If any player $j\notin \CR(v,\pi)$ gets a positive share, she can not pay the value back to those critical players as no marginal contribution is created when she is the last to join. Therefore, the value should be shared only among critical players.
\end{itemize}
\end{proof}

\begin{proof}[Proof of Proposition~\ref{prop: MOS=MOC}]
    
    ``$\Rightarrow$'': Given 0-1 monotone $v$ and an order $\pi$, notice that if $i,j$ are symmetric, then $i\in\CR(v,\pi)\Leftrightarrow j\in\CR(v,\pi)$. Suppose $i\prec_\pi j$, if both $i,j$ are critical players, MOS is guaranteed when the share of $i$ is higher than $j$; if $i,j$ are not critical players, their share should be $0$ following SF and OIR.
    
    ``$\Leftarrow$'': If the share of $i$ is less than $j$, then the local game on $\previous{\pi}{i^*}$ is a counter example for MOS.
\end{proof}

\subsection{WVS is SF, OIR, I4EA and MOS}

Before the proof of Theorem~\ref{thm: WVS complete}, we emphasize that WVS can verify if a game is solvable properly. WVS firstly determines the share of the marginal player, and shares the rest of the value among other critical players. The share is inefficient only when (1) the share of the marginal player is less than $1$ and (2) the marginal player is the unique critical player. Therefore, there exists a joining order where the marginal player is not the unique critical player, and she can delay to be the unique one. Such cases are unsolvable following Definition~\ref{def: solvable}. Now, we give the proof of Theorem~\ref{thm: WVS complete} formally.

\begin{proof}[Proof of Theorem~\ref{thm: WVS complete}]

\textbf{OIR}. WVS is obviously OIR as the value is shared only once when the marginal player joins, and the share is unchanged in the future.

\textbf{SF}.  Given $N$ and $i\in N$, we construct a partition $\mathcal{P}(i,\Pi)=\{\Pi_{i0},\Pi_{i1},...,\Pi_{im}\}$ of the joining order set $\Pi(N)$ as follows:
    \begin{itemize}
        \item for any $\pi$ where $i\notin\CR(v,\pi)$, $\pi\in\Pi_{i0}$;
        \item for every $1\leq k \leq m$, there exists only one $\pi\in \Pi_{ik}$, where $i$ is the marginal player of $\pi$;
        \item for $\Pi_{ik}$ where $1\leq k \leq m$, if $\pi\in \Pi_{ik}$, then $\pi'\in \Pi_{ik}$. Here $\pi'$ is the joining order where $i$ exchanges her position with each one previous critical player in $\pi$.
    \end{itemize}
    Notice that for any $\pi$ such that $i\in \CR(v,\pi)$ while $i$ is not the marginal player, it corresponds to another $\pi'$ where $i$ exchanges the position with the marginal player in $\pi$ to become the marginal player in $\pi'$ without changing the critical player set. Hence, the above partition finally includes all $\pi\in\Pi(N)$. As we mentioned, the critical players are symmetric with each other in the local game $(\localv[\previous{\pi}{i}],\previous{\pi}{i})$. Therefore, in an anonymous mechanism for $\pi=[...,j,...,i,...]$ and $\pi'=[...,i,...,j,...]$ satisfying $\{i,j\}\subseteq \CR(v,\pi) = \CR(v,\pi')$, we have $\share{i}{\mechanism}{v}{\pi}=\share{j}{\mechanism}{v}{\pi'}$. Notice that $\sum_{j\in \CR(v,\pi)}\share{i}{\WVS}{v}{\pi}=1$ according to the efficiency, we have 
    $\sum_{\pi\in \Pi_{ik}}\share{i}{\WVS}{v}{\pi}=1.$
    Now consider the total marginal contribution produced by $i$ in the joining orders in $\Pi_{ik}$. For $1\leq k \leq m$, as there is only one order that $i$ is the marginal player, we have 
    $\sum_{\pi\in\Pi_{ik}}\MC_i(v,\pi)=1$. For $k=0$, $i$ is never the marginal player so the total marginal contribution is $0$. Hence, 
    $ \sum_{\pi\in \Pi_{ik}}\share{i}{\WVS}{v}{\pi} = \sum_{\pi\in\Pi_{ik}}\MC_i(v,\pi) $
    Overall,
    \begin{align*}
        &\frac{1}{n!}\sum_{\Pi_{ik}\in \mathcal{P}(i,\Pi)}\sum_{\pi\in \Pi_{ik}}\share{i}{\WVS}{v}{\pi}\\
        =&\frac{1}{n!}\sum_{\Pi_{ik}\in \mathcal{P}(i,\Pi)}\sum_{\pi\in\Pi_{ik}}\MC_i(v,\pi)\\
        =& \frac{1}{n!}\sum_{\pi\in \Pi(N)} \MC_i(v,\pi) =\SV_i(v).
    \end{align*}    
\textbf{I4EA}. For any solvable game $v$, and $\pi=[...,j,i...]$,$\pi'=[...i,j...]$, we discuss the share of $i$ by categories:
\begin{itemize}
    \item If $i\notin \CR(v,\pi)$, then we can infer $i\notin \CR(v,\pi')$, so $\share{i}{\WVS}{v}{\pi'}=\share{i}{\WVS}{v}{\pi}=0$. 
    \item If $i\in \CR(v,\pi)$ but not the marginal player, we first show that $\CR(\localv[\mcp(v,\pi)],\localpi[\mcp(v,\pi)]) \subseteq \CR(\localv[\mcp(v,\pi')],\localpi[\mcp(v,\pi')])$. We denote the marginal player of $\pi$ as $i^*$, then $i^*$ is also the marginal player of $\pi'$ following Observation~\ref{ob: CR unchanged}. Notice that, if $i^*$ is moved one position forward simultaneously in $\pi$ and $\pi'$, we still have $\previous{\pi}{i} = \previous{\pi'}{i}$, and the difference only happens when $i^*$ is the next player of $i$ in $\pi$. As $i\in\CR(v,\pi)$, we have $i\in \mcp(v,\pi)$, which means $i^*$ would not be moved before $i$ when finding $\mcp(v,\pi)$. In $\pi'$, $i$ is also included in $\mcp(v,\pi)$, but $j$ may be excluded. Therefore, $\mcp(v,\pi')\subseteq \mcp(v,\pi)$, so $\CR(\localv[\mcp(v,\pi)],\localpi[\mcp(v,\pi)]) \subseteq \CR(\localv[\mcp(v,\pi')],\localpi[\mcp(v,\pi')])$ following Observation~\ref{ob: CR decrease or increase}. This suffices to show $\share{i^*}{\WVS}{v}{\pi}\geq \share{i^*}{\WVS}{v}{\pi'}$. As  we have $\CR(v,\pi)=\CR(v,\pi')$ following Observation~\ref{ob: CR unchanged} and $i$ arrives earlier in $\pi'$, $\share{i}{\WVS}{v}{\pi'}\geq \share{i}{\WVS}{v}{\pi}.$

    % Recall that the share of $i^*$ is determined by recursively moving $i^*$ one position forward, and the orders of players before $i^*$ are the same in $\pi,\pi'$ except for $i,j$. Consider the recursion in $\pi$: (1) If the recursion stops before $i^*$ is the next player of $i$, then the recursion in $\pi'$ also stops at the same position. (2) If not, then the recursion in $\pi$ must stop when $i^*$ is the next player of $i$ because $i\in\CR(v,\pi)$. Then, the recursion in $\pi'$ stops when $i^*$ is the next player of $j$ or $i^*$ is between $i$ and $j$, while the former case means that $j$ is also not critical in the recursion in $\pi$. Overall, when the recursions in $\pi,\pi'$ both stop, the critical players are the same, so the value share of $i^*$ is same. Back to $\pi,\pi'$, we have $\share{i^*}{\EGM}{v}{\pi}=\share{i^*}{\EGM}{v}{\pi'}$, so $\share{i}{\EGM}{v}{\pi}=\share{i}{\EGM}{v}{\pi'}$.
    \item If $i$ is the marginal player of $\pi$, but not the marginal player of $\pi'$, then $j$ is the marginal player of $\pi'$, and $\CR(v,\pi)=\CR(v,\pi')$. As we mentioned, $i$'s share in $\pi$ is determined by $\CR(\localv[\mcp(v,\pi)],\localpi[\mcp(v,\pi)])$ and $|\CR(\localv[\mcp(v,\pi)],\localpi[\mcp(v,\pi)])|\geq |\CR(v,\pi)|$. Note that $i$ is the last player in $\CR(\localv[\mcp(v,\pi)],\localpi[\mcp(v,\pi)])$, but the penultimate player in $\CR(v,\pi')$, so $\share{i}{\WVS}{v}{\pi}\leq \share{i}{\WVS}{v}{\pi'}$.
    \item If $i$ is the marginal player of $\pi,\pi'$, then following the recursively determined share of $i$, we have $\share{i}{\EGM}{v}{\pi}=\share{i}{\EGM}{v}{\pi'}$ except for the case $\CR(v,\pi)=1<\CR(v,\pi')$. However, as we mentioned, this would only happen when $v$ is unsolvable.
\end{itemize}
\textbf{MOS}.  Given solvable $v$ and $\pi$, let $i$ be the marginal player and $j$ be the penultimate critical player, the proof is equivalent to show $i$'s share is weakly less than $j$'s. Recall Observation~\ref{ob: CR decrease or increase}, we have $\CR(v,\pi)\subseteq\CR(\localv[\mcp(v,\pi)],\localpi[\mcp(v,\pi)])$so the weight of $i$ is less than $j$ and the total weight determining $i$'s share is larger than $j$'s, which is sufficient for the proof.
\end{proof}

\subsection{EVS Maximizes EW}
We prove that EVS maximizes the EW comparing to all mechanisms satisfying SF, OIR, I4EA and MOS.

\begin{proof}[Proof of Theorem~\ref{thm: EGM max EW}]
    
    Let $\pi=[...,j,i,...]$ and $\pi'=[...,i,j,...]$ where $i$ is the marginal player of $\pi$. Following the constraint of I4EA, we require 
    $\share{i}{\mechanism}{v}{\pi}\leq \share{i}{\mechanism}{v}{\pi'}$. Now we discuss by categories as follows: \begin{itemize}
        \item If $i$ is not the marginal player of $\pi'$, we point out that $\EGM$  maximized $\share{i}{}{v}{\pi}$ as all critical players have the same share in $(v,\pi)$. If one mechanism assigns a strictly higher share of $i$, then there exists $i'\in \CR(v,\pi)\setminus\{i\}$ that $\share{i'}{}{v}{\pi}<\share{i}{}{v}{\pi}$, which leads to a contradiction.
        \item If $i$ is the marginal player of $\pi'$, then $\share{i}{}{v}{\pi}$ is maximized by $\EGM$ as it assigns $\share{i}{\mechanism}{v}{\pi} = \share{i}{\mechanism}{v}{\pi'}$.
    \end{itemize}
    Therefore, $\EGM$ maximizes the share of the marginal player recursively, so it maximizes the $\EW$.   
\end{proof}

\subsection{EVS Minimizes Expected SD}

We prove that EVS minimizes the expected SD comparing to all mechanisms satisfying SF, OIR, I4EA and MOS.

\begin{proof}[Proof of Theorem~\ref{thm: EGM min d to SV}]
    Recall the partition $\mathcal{P}(i,\Pi)=\{\Pi_{i0},\Pi_{i1},...,\Pi_{im}\}$ constructed in the Proof of Theorem~\ref{thm: WVS complete} and $\sum_{\pi\in \Pi_{ik}}\share{i}{\mechanism}{v}{\pi}=1$ where $\mechanism$ is SF, OIR, I4EA and MOS. 
    Now, consider the expected $\SD$ of a mechanism $\mechanism$:
    \begin{align*}
        \expectation[\pi]{\SD^\mechanism(v,\pi)}=\frac{1}{n!}\sum_{\pi\in\Pi(N)}\sum_{i\in \pi} (\share{i}{\mechanism}{v}{\pi}-\SV_i(v))^2 \\
        =\frac{1}{n!}\sum_{i\in N}\sum_{\Pi_{ik}\in\mathcal{P}(i,\Pi)}\sum_{\pi\in\Pi_{ik}}(\share{i}{\mechanism}{v}{\pi}-\SV_i(v))^2. \\
    \end{align*}
    The innermost term $\sum_{\pi\in\Pi}(\share{i}{\mechanism}{v}{\pi}-\SV_i(v))^2$ is minimized by EVS. Since $\sum_{\pi\in \Pi_{ik}}\share{i}{\mechanism}{v}{\pi}=1$ and $\SV_i$ is fixed, this term only varies with $\sum_{\pi\in\Pi_{ik}}(\share{i}{\mechanism}{v}{\pi})^2$, i.e., it can be written as
    \[ -2\SV_i(v)+\sum_{\pi\in\Pi_{ik}}\left((\share{i}{\mechanism}{v}{\pi})^2+(\SV_i(v))^2\right) \]
    which is minimized when $\{\share{i}{\mechanism}{v}{\pi}\}_{\pi\in\Pi_{ik}}$ are as equal as possible. Following the requirements of MOS and anonymous, the minimal term in $\{\share{i}{\mechanism}{v}{\pi}\}_{\pi\in\Pi_{ik}}$ is the share when $i$ is the marginal player in $\pi$. Since EVS maximizes the egalitarian share, and other shares in $\{\share{i}{\mechanism}{v}{\pi}\}_{\pi\in\Pi_{ik}}$ are equal to each other, the term $\sum_{\pi\in\Pi}(\share{i}{\mechanism}{v}{\pi}-\SV_i(v))^2$ is minimized. Therefore, EVS minimized expected SD.
    
    % EGM realizes this minimization by maximizing the egalitarian welfare, which is also the 
    
    % Notice that in $\EGM$, $\share{i}{\EGM}{v}{\pi}$ is maximized when $i$ is the marginal player, and for $\pi_1,\pi_2\in \Pi_{ik}$, we have $\share{i}{\EGM}{v}{\pi_1}=\share{i}{\EGM}{v}{\pi_2}$ when $i$ is not the marginal player in $\pi_1$ and $\pi_2$. Therefore, $\EGM$ minimizes expected $\SD$ by minimizing $\sum_{\pi\in\Pi}(\share{i}{\mechanism}{v}{\pi}-\SV_i(v))^2$ under the constraint of $\sum_{\pi\in \Pi_{ik}}\share{i}{\mechanism}{v}{\pi}=1$.    
\end{proof}

\section{Acknowledgments}

This work was supported in part by the Science and Technology Commission of Shanghai Municipality (No. 23010503000), the Shanghai Frontiers Science Center of Human-centered Artificial Intelligence (ShangHAI), and JST ERATO (Grant Number JPMJER2301, Japan).

\bibliography{aaai2026}

\end{document}